\documentclass[11pt]{article}

\usepackage{amsfonts}
\usepackage{mathrsfs}
\usepackage{amssymb}
\usepackage{stmaryrd}
\usepackage{amsmath}
\usepackage{graphicx}
\usepackage{subfig}
\usepackage{arydshln}
\usepackage{color}
\usepackage{ifpdf}

\newtheorem{remark}{Remark}
\newtheorem{theorem}{Theorem}
\newtheorem{lemma}{Lemma}

\newenvironment{proof}{\makebox[7ex][l]{\it Proof:\/}}{\hfill\/ \hfill\/ {\it
Q.E.D.} \vspace{0.5ex}\\}

\begin{document}

\title{\LARGE \bf
Hybrid Control of ADT Switched Linear Systems subject to Actuator Saturation
}

\author{Fen Wu\thanks{Email: {\tt fwu@ncsu.edu}, Tel: (919) 515-5268.}
\\
Department of Mechanical and Aerospace Engineering \\
North Carolina State University \\
Raleigh, NC 27695, USA
\vspace*{0.1in} \\
Chengzhi Yuan \\
Department of Mechanical, Industrial and Systems Engineering \\
University of Rhode Island \\
Kingston, RI 02881, USA
}

\date{February 20, 2026}

\maketitle

\begin{abstract}
This paper develops a hybrid output-feedback control framework for average dwell-time (ADT) switched linear systems subject to actuator saturation.
The considered subsystems may be exponentially unstable, and the saturation nonlinearity is explicitly handled through a deadzone-based representation.
The proposed hybrid controller combines mode-dependent full-order dynamic output-feedback controllers with a supervisory reset mechanism that updates controller states at switching instants.
By incorporating the reset rule directly into the synthesis conditions, switching boundary constraints and performance requirements are addressed in a unified convex formulation.
Sufficient conditions are derived in terms of linear matrix inequalities (LMIs) to guarantee exponential stability under ADT switching and a prescribed weighted ${\cal L}_2$-gain disturbance attenuation level for energy-bounded disturbances.
An explicit controller construction algorithm is provided based on feasible LMI solutions.
Simulation results demonstrate the effectiveness and computational tractability of the proposed approach and highlight its advantages over existing output-feedback saturation control methods.
\end{abstract}



\section{Introduction}
\label{Sec.Int}

Recent years have witnessed significant advances in hybrid control systems, driven by both theoretical interest and practical demand.
Hybrid systems have been intensively studied in the control community because they provide effective modeling frameworks for complex dynamics and systems with large uncertainties, and because hybrid control structures can relax or overcome inherent limitations of traditional control designs, thereby improving performance and design flexibility (see, e.g., \cite{BraBM.TAC98,Ant.IEEE00,McCK.IEEE00,SchS.B00,GoeST2012} and the references therein).

Switched linear systems constitute an important subclass of hybrid systems.
They consist of a finite family of linear subsystems, continuous- or discrete-time, together with a switching logic governing transitions among them, and have been widely investigated in the literature \cite{WicksPD.EJC98,Bra.TAC98,LiberzonB03,SunG.B05,LinA.TAC09}. Despite the availability of powerful linear analysis tools, stability and stabilization of switched linear systems remain challenging, since instability may arise even when all subsystems are individually asymptotically stable \cite{LibM.CSM99}.
As a result, the switching signal plays a critical role in determining system behavior and must be explicitly addressed in controller design.

For autonomous switching, stability analysis can be cast as a robust problem, where the existence of a common Lyapunov function guarantees stability under arbitrary switching. However, such conditions are often conservative, particularly when specific switching mechanisms are enforced. In contrast, for controlled switching, it is natural to incorporate the switching logic into the design process. The multiple Lyapunov function (MLF) framework enables this integration and underpins several widely used switching strategies, including min/max switching, hysteresis switching, and dwell-time or average dwell-time switching \cite{Mor.B97,Bra.TAC98,YeMH.TAC98,LibM.CSM99}. Among these, average dwell-time logic offers a flexible and practically relevant balance between performance and stability guarantees \cite{LiberzonB03}.

For state-feedback control, quadratic MLF-based approaches lead to convex LMI formulations and efficient numerical synthesis \cite{BoyGFB.B04}. However, output-feedback synthesis under average dwell-time constraints is considerably more challenging, as switching-induced boundary conditions introduce non-convex bilinear matrix inequalities (BMIs) \cite{LuW.Au04,Wu2001}. Existing approaches either decouple controller design from dwell-time constraints \cite{WanZDL.IJRNC09}, leading to conservatism, or rely on controller-state reset mechanisms that effectively require state feedback \cite{LuW.TCST06}.
More recently, hybrid output-feedback frameworks incorporating controller-state resets have enabled convex LMI-based synthesis while explicitly accounting for switching boundary conditions, significantly reducing conservatism and improving achievable performance \cite{YuaW2015,YuaLWD2016}.

Actuator saturation is an inherent nonlinearity in many engineering systems and can significantly degrade closed-loop performance or even destabilize an otherwise stable control system. As a result, saturation control has been extensively studied in the control community, addressing fundamental problems such as stabilization, output regulation, and disturbance attenuation \cite{BerM95,TarG97,HuL2001}.
Early work primarily focused on open-loop stable linear time-invariant (LTI) systems, for which global or semi-global stabilization under actuator saturation can often be achieved using well-established analytical and synthesis techniques
\cite{SusSY94,LinST96,LiuCS96,SabLT96}.

In contrast, many practical control systems—most notably in aerospace, robotics, and high-performance motion control—are open-loop exponentially unstable and subject to severe actuator limitations. For such systems, actuator saturation fundamentally alters the system’s controllability properties: global stabilization is no longer achievable, and control objectives must be reformulated in terms of null controllability over a bounded region of the state space
\cite{NguJ99,HuL2001,PaiTGC2002,HuTZ2006}.
Consequently, a central research focus has been the characterization of null controllable regions and the development of feedback laws that ensure stability and performance over the entire region or a large subset thereof. Numerous design methodologies grounded in rigorous nonlinear and Lyapunov-based analysis have been proposed for this class of systems.

Despite these advances, the vast majority of results for exponentially unstable systems rely on full state feedback. Output-feedback saturation control remains significantly less developed, and existing approaches often introduce auxiliary deadzone or gain-scheduling structures.
As a result, the associated synthesis conditions typically involve bilinear matrix inequalities (BMIs) or lead to complex controller parameterizations with high computational complexity  \cite{NguJ99,ScoFE2002,WuLZ2007,WuZL2009,DaiHTZ2009,BanW2015}.
In \cite{WuLZ2007,WuZL2009}, a gain-scheduled output-feedback saturation control approach was developed to locally stabilize saturated linear systems while attenuating disturbance effects on the system output.
Alternatively, \cite{DaiHTZ2009} proposed an  output-feedback saturation control synthesis method based on an internal deadzone loop, for which the synthesis conditions can be formulated as LMIs.
Despite their effectiveness, these methods are largely restricted to single LTI systems.
To date, systematic and tractable output-feedback control synthesis for switched linear systems with actuator saturation under average dwell-time (ADT) switching has not been adequately addressed \cite{ZhoY2015,CheM2021}.
In contrast, output-feedback control for switched linear systems employing minimum switching logic was considered in \cite{DuaW2016}.


The objective of this research is to develop a systematic and tractable output-feedback control framework for ADT switched linear systems subject to actuator saturation.
Building on hybrid control methodologies for switched systems, the proposed approach aims to explicitly account for saturation nonlinearities while preserving stability and performance guarantees.
In particular, this work seeks to integrate saturation-aware control design with hybrid switching mechanisms, enabling the synthesis of output-feedback controllers that ensure stability over a well-characterized region of attraction for exponentially unstable subsystems.
By embedding switching-induced boundary conditions and actuator constraints directly into the control design, the proposed framework targets practical applicability in systems where both switching and input limitations are intrinsic.

The main contributions of this work are threefold.
First, a hybrid output-feedback control framework is developed for ADT switched linear systems with actuator saturation, explicitly incorporating both saturation nonlinearities and switching boundary conditions into the synthesis process.
Second, sufficient stability and performance conditions are derived using a multiple Lyapunov function approach, leading to convex LMI-based synthesis conditions that avoid the bilinear matrix inequalities commonly encountered in existing output-feedback saturation control methods.
Third, the proposed framework unifies and extends existing results on hybrid control and saturation control, thereby providing a systematic solution to a class of problems that has not been adequately addressed in the literature.
The effectiveness and reduced conservatism of the proposed approach are demonstrated through representative numerical examples.

The notation used throughout this paper is standard.
The set of real numbers is denoted by $\mathbb{R}$, and  $\mathbb{R}_+$ denotes the set of positive real numbers.
The space of real $m \times n$ matrices is denoted by $\mathbb{R}^{m\times n}$.
For a real matrix $M$, its transpose is denoted by $M^T$,
and the hermitian operator ${\it He} \{\cdot\}$ is defined as
${\it He} \{M\} = M + M^T$.
The identity matrix of appropriate dimension is denoted by $I$.
The set of real symmetric $n \times n$ matrices is denoted by $\mathbb{S}^{n\times n}$, and $\mathbb{S}^{n\times n}_+$ denotes the subset of positive definite matrices.
For $M \in \mathbb{S}^{n\times n}$, the notation $M > 0$ $(M
\geq 0)$ indicates that $M$ is positive definite (positive semi-definite), while $M < 0$ $(M \leq 0)$ denotes negative
definite (negative semi-definite).
In LMIs, the symbol $\star$ is used to represent entries
that are implied by symmetry.
For a vector $x \in \mathbb{R}^{n}$, the Euclidean norm is
defined as $\|x\| = (x^T x)^{1/2}$.
The space of square-integrable functions is denoted by
$\mathcal{L}_2$; that is, for any $u\in \mathcal{L}_2$, $\| u
\|_2 = [\int^{\infty}_{0} u^T(t) u(t) dt]^{1/2}$ is finite.
Finally, for two integers $k_1 < k_2$, the index set ${\bf
I}[k_1,k_2] = \{k_1, k_1+1, \cdots, k_2\}$.

The remainder of this paper is organized as follows.
Section \ref{Sec.Pre} presents preliminary results on the stability and weighted $\mathcal{L}_2$-gain properties of switched linear systems under average dwell-time switching and introduces tools for handling actuator saturation.
Section \ref{Sec.Main} formulates the output-feedback control problem for ADT switched linear systems with actuator saturation and develops LMI-based synthesis conditions for full-order hybrid controllers.
Section \ref{Sec.Sim} provides numerical examples and simulation results to demonstrate the effectiveness of the proposed approach.
Finally, Section \ref{Sec.Con} concludes the paper.

\section{Preliminaries}
\label{Sec.Pre}

This section reviews several existing results on switched linear systems and saturations that are instrumental for the subsequent development.

Consider the continuous-time switched linear system
\begin{equation}
\label{AutoPlant}
\begin{bmatrix}
\dot{x} \\
z
\end{bmatrix} = \begin{bmatrix}
A_{\sigma} & B_{\sigma} \\
C_{\sigma} & D_{\sigma}
\end{bmatrix} \begin{bmatrix}
x \\
w
\end{bmatrix},
\end{equation}
where $x \in \mathbb{R}^n$ denotes the system state with initial condition $x(0)
= x_0$, $z \in \mathbb{R}^{n_z}$ is the controlled output, $w \in
\mathbb{R}^{n_w}$ represents an exogenous disturbance.
The switching signal $\sigma(t): \ [0,\infty)
\rightarrow {\bf I}[1,N]$ is a piecewise constant function selecting one of $N > 1$ subsystems.
For each mode $i \in \mathbf{I}[1,N]$, the matrices
$A_\sigma, B_\sigma, C_\sigma, D_\sigma$ are constant and
of compatible dimensions.

The switched system (\ref{AutoPlant}) is said to satisfy an
average dwell-time switching logic if, for any switching
signal $\sigma(t)$,
there exist two constants $N_0 > 0$ and $\tau_a > 0$ such that
\begin{equation}
\label{dwelltime scheme}
N_\sigma(T,t) \leq N_0 + \frac{T-t}{\tau_a}, \qquad \forall \ 0 \leq t \leq T,
\end{equation}
where $N_\sigma(T,t)$ denotes the number of switches of
$\sigma(t)$ over the interval $(t,T)$.
Here, $N_0$ is the chatter bound and
$\tau_a$ is the average dwell-time \cite{LiberzonB03}.
This concept generalizes classical dwell-time switching by allowing temporary fast switching, compensated by
sufficiently slow switching thereafter
\cite{Mor.TAC96,LiberzonB03}.

When $w \equiv 0$, the switched system (\ref{AutoPlant}) is
said to be globally exponentially stable with convergence
rate $\lambda > 0$ under a given switching signal $\sigma(t)$
if $\|x(t)\| \leq e^{-\lambda (t-t_0)} \|x(t_0)\|, \ \forall \ t \geq t_0$.
The following result provides a sufficient condition for exponential stability under average dwell-time switching.

\begin{lemma}[\cite{LibM.CSM99}]
\label{lemma1}
Consider the switched nonlinear system
\begin{equation}
\label{NonPlant}
\dot{x}(t) = f_i(x(t)), \qquad i \in {\bf I}[1,N].
\end{equation}
Suppose there exist Lyapunov-like functions $V_i$,
positive constants $a_i, b_i$, and scalars
$\lambda_0 > 0$ and $\mu > 1$ such that, for all $x\in\mathbb{R}^n$ and all $i,j \in \mathbf{I}[1,N]$,
\begin{align}
& a_i\|x\|^2 \leq V_i(x) \leq b_i \|x\|^2 \\
& \frac{\partial V_i(x)}{\partial x} f_i(x) \leq -\lambda_0 \|x\|^2 \\
& V_i(x) \leq \mu V_j(x). \label{Boundary}
\end{align}
Then system (\ref{NonPlant}) is globally exponentially
stable for every switching signal with average dwell time
$\tau_a \geq \frac{\ln(\mu)}{\lambda_0}$.
\footnote{Note that the average dwell time $\tau_a$ is explicitly
determined by the values of $\lambda_0$ and $\mu$.
For the convenience of presentation, these two parameters $\lambda_0$
and $\mu$ will be called as the ``dwell-time parameters'' in the
following development.}
\end{lemma}

Finally, the switched system (\ref{AutoPlant}) is said to
have a weighted $\mathcal{L}_2$-gain $\gamma > 0$ under
average dwell-time switching if there exist $\lambda > 0$
and a class $\mathcal{K}_\infty$ function $\beta(\cdot)$ such that
\begin{equation}
\label{weighted L2}
\int^{\infty}_0 e^{-\lambda t} z^T(t) z(t) dt \leq \beta(x(0))
+ \gamma^2 \int^{\infty}_0 w^T(t) w(t) dt.
\end{equation}

%

To facilitate the treatment of actuator saturation nonlinearities, we employ the following result adapted from \cite{HuTZ2006}, which provides a convenient convex representation of saturation and deadzone nonlinearities and will be instrumental in the subsequent analysis.

\begin{lemma}[\cite{HuTZ2006}]
\label{lem:poly}
Let $h(x) = H x$ be a linear map and suppose $\ell_j^T H x
\in \left[ -\bar{u}_j, \bar{u}_j \right]$, where $\ell_j$
denotes the $j$th column of the identity matrix.
Then for any scalar $u_j$, the saturation nonlinearity
satisfies ${\rm sat}(u_j) \in {\rm Co} \left\{ u_j,
\ell_j^T H x \right\}$, and and the associated deadzone nonlinearity can be expressed as
${\rm dz}(u_j) = \delta (u_j - \ell_j^T H x)$,
for some $\delta \in [0, 1]$.
\end{lemma}


\section{Hybrid Control Synthesis}
\label{Sec.Main}

\subsection{Problem Statement}
Consider a switched linear plant subject to actuator saturation described by
\begin{equation}
\label{plant}
\begin{bmatrix}
\dot{x}_p \\
z \\
y
\end{bmatrix} = \begin{bmatrix}
A_{p,i} & B_{p1,i} & B_{p2,i} \\
C_{p1,i} & D_{p11,i} & D_{p12,i} \\
C_{p2,i} & D_{p21,i} & D_{p22,i}
\end{bmatrix}
\begin{bmatrix}
x_p \\
w \\
{\rm sat}(u)
\end{bmatrix},
\end{equation}
where $i \in \mathbf{I}[1,N_p]$ indexes the active subsystem and $N_p > 1$ denotes the total number of subsystems.
The state $x_p \in \mathbb{R}^n$ represents the plant state,
$u \in \mathbb{R}^{n_u}$ is the control input, $z \in \mathbb{R}^{n_z}$ is the controlled output, $w \in \mathbb{R}^{n_w}$ is an exogenous disturbance, and $y \in \mathbb{R}^{n_y}$ is the measured output available for
control use. All system matrices are assumed to be known and of compatible dimensions.

The function ${\rm sat}(\cdot)$ denotes a componentwise saturation nonlinearity with saturation limits $\bar{u}_j >
0$, $j \in \mathbf{I}[1,n_u]$, defined as
\begin{align*}
{\rm sat}(u_j) = \begin{cases}
   \bar{u}_j, \quad & u_j \geq \bar{u}_j \\
   u_j, & u_j \in (-\bar{u}_j, \bar{u}_j) \\
   -\bar{u}_j, & u_j \leq -\bar{u}_j
\end{cases}.
\end{align*}
The following standard assumptions are imposed:
\begin{description}
\item[(A1)]
The triple $(A_{p,i},B_{p2,i},C_{p2,i})$ is stabilizable and
detectable for all $i \in {\bf I}[1,N_p]$.
\item[(A2)]
$D_{p22,i} = 0$ for all $i\in {\bf I}[1,N_p]$.
\end{description}
Assumption (A1) guarantees the existence of a stabilizing dynamic output-feedback controller for each subsystem, while (A2) simplifies the controller structure and can be relaxed
via loop transformation.

This paper investigates the output-feedback control synthesis problem for the ADT switched linear system \eqref{plant} in the presence of actuator saturation, with the objective of stabilizing the switched system and minimizing the effect of disturbances.
The dwell-time parameters $\lambda_0$ and $\mu$ associated
with the switching boundary condition (\ref{Boundary}) are assumed to be pre-specified.
Owing to these boundary conditions and actuator saturation, conventional synthesis approaches typically lead to non-convex optimization problems.
The goal here is to develop a hybrid control framework that converts the synthesis problem into a convex LMI-based optimization, enabling efficient numerical solution.

For disturbance attenuation, we consider energy-bounded disturbances belonging to the set
\[
{\cal W}_s  = \left\{ w: \ {\bf R}_+ \rightarrow {\bf R}^{n_w},
\int_{0}^\infty \ w^T(\tau) w(\tau) d\tau \leq s^2, \ w \in
{\cal L}_2 \right\},
\]
where $s > 0$ specifies the disturbance energy level.
The disturbance attenuation performance is quantified by the weighted ${\cal L}_2$ gain, as defined in (\ref{weighted L2}).

To this end, a full-order hybrid dynamic output-feedback controller is constructed for system \eqref{plant}, with each controller mode corresponding to a plant subsystem:
\begin{align}
\begin{aligned}
\label{Controller}
\begin{bmatrix}
\dot{x}_k \\
u
\end{bmatrix} &= \begin{bmatrix}
A_{k,i} & B_{k1,i} & B_{k2,i} \\
C_{k1,i} & D_{k11,i} & D_{k12,i}
\end{bmatrix}
\begin{bmatrix}
x_k \\
y \\
{\rm dz}(u)
\end{bmatrix} \\
x_k^+ &= \Delta_{ij} x_k, ~\mbox{when switching occurs}
\end{aligned}
\end{align}
where $x_k \in \mathbb{R}^n$ denotes the controller state.
The reset matrix $\Delta_{ij}$ maps the controller state at
the switching instant from the pre-switching mode $i$ to
the post-switching mode $j$.

The proposed hybrid control architecture consists of a switching dynamic output-feedback controller together with a supervisory mechanism that enforces controller-state resets induced by the embedded hybrid loop (illustrated by the gray block $\mathcal{H}$ in Fig.~\ref{Fig.Hyb_sat}).
During flow intervals, the plant is regulated by the active controller, while at switching instants, the controller state is reset according to the prescribed reset law, thereby explicitly accounting for switching-induced boundary conditions.

\begin{figure}[!htb]
\centering
\includegraphics[scale=0.35]{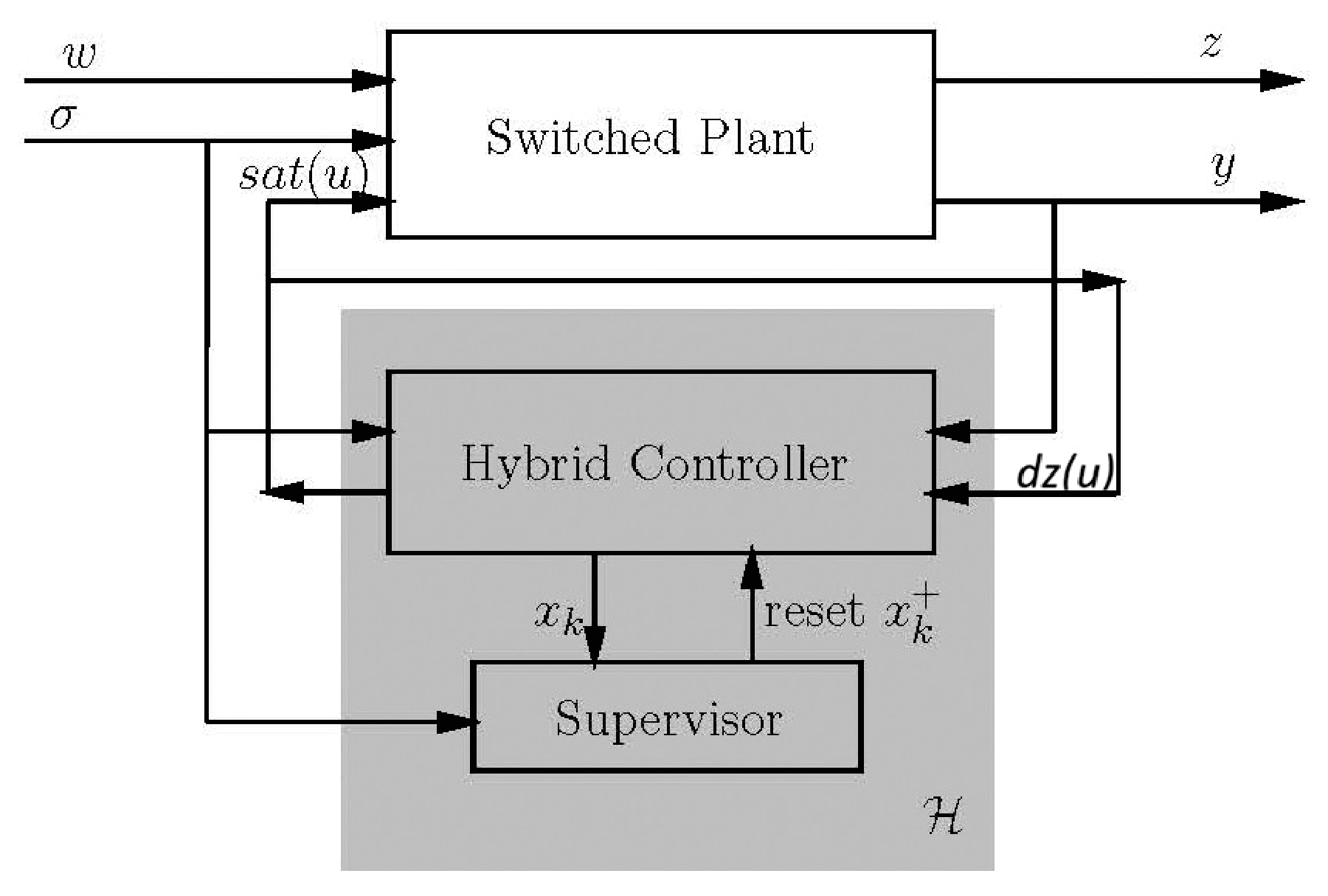}
\caption{The proposed hybrid control scheme for saturated ADT switched systems.}
\label{Fig.Hyb_sat}
\end{figure}

\subsection{Hybrid Control synthesis condition}

In the following development, a hybrid control framework is proposed for the ADT switched linear system with saturation \eqref{plant} using full-order dynamic output-feedback controllers \eqref{Controller}.

To explicitly handle actuator saturation, we introduce the deadzone nonlinearity ${\rm dz}(u)=u-{\rm sat}(u)$.
Using this representation, the state-space equations of the plant \eqref{plant} can be equivalently rewritten in the following linear fractional transformation (LFT) form:
\begin{align}
\label{eqn:LFTPlant}
\begin{aligned}
\begin{bmatrix}
\dot{x} \\
u \\
z \\
y
\end{bmatrix} & = \begin{bmatrix}
A_{p,i} & -B_{p2,i} & B_{p1,i} & B_{p2,i} \\
0 & 0 & 0 & I \\
C_{p1,i} & -D_{p12,i} & D_{p11,i} & D_{p12,i} \\
C_{p2,i} & 0 & D_{p21,i} & 0
\end{bmatrix} \begin{bmatrix}
x \\
p_u \\
w \\
u
\end{bmatrix} \\
p_u & = {\rm dz}(u)
\end{aligned}
\end{align}
where $p_u$ denotes the deadzone signal associated with the saturated control input.

By interconnecting the switched plant \eqref{plant} with the dynamic output-feedback controller \eqref{Controller}, the closed-loop hybrid system can be expressed as
\begin{align}
\begin{bmatrix}
\dot{x}_{cl} \\
u \\
z
\end{bmatrix} &= \begin{bmatrix}
A_{cl,i} & B_{cl0,i} & B_{cl2,i} \\
C_{cl0,i} & D_{cl00,i} & D_{cl02,i} \\
C_{cl2,i} & D_{cl20,i} & D_{cl22,i}
\end{bmatrix}
\begin{bmatrix}
x_{cl} \\
p_u \\
w
\end{bmatrix} \label{Closed1} \\
p_u &= {\rm dz}(u) \\
\begin{bmatrix}
x_p^+ \\
x_k^+
\end{bmatrix} &= \begin{bmatrix}
I & 0 \\
0 & \Delta_{ij}
\end{bmatrix}
\begin{bmatrix}
x_p \\
x_k
\end{bmatrix} \nonumber \\
&:= A_{s,ij} x_{cl}, ~\mbox{when switching occurs} \label{Closed2}
\end{align}
where $x_{cl} = \begin{bmatrix} x_p^T & x_k^T \end{bmatrix}^T$ is the augmented closed-loop state.
The closed-loop system matrices are given by
\begin{align*}
& \begin{bmatrix}
A_{cl,i} & B_{cl0,i} & B_{cl2,i} \\
C_{cl0,i} & D_{cl00,i} & D_{cl02,i} \\
C_{cl2,i} & D_{cl20,i} & D_{cl22,i}
\end{bmatrix}
= \left[ \begin{array}{cc:c:c}
A_{p,i} & 0 & -B_{p2,i} & B_{p1,i} \\
0 & 0 & 0 & 0 \\ \hdashline
0 & 0 & 0 & 0 \\ \hdashline
C_{p1,i} & 0 & -D_{p12,i} & D_{p11,i}
\end{array} \right] \\
& \hspace*{0.5in} + \left[ \begin{array}{cc}
0 & B_{p2,i} \\
I & 0 \\ \hdashline
0 & I \\ \hdashline
0 & D_{p12,i}
\end{array} \right]
\begin{bmatrix}
A_{k,i} & B_{k2,i} & B_{k1,i} \\
C_{k1,i} & D_{k12,i} & D_{k11,i}
\end{bmatrix}
\left[ \begin{array}{cc:c:c}
0 & I & 0 & 0 \\
0 & 0 & I & 0 \\
C_{p2,i} & 0 & 0 & D_{p21,i}
\end{array} \right].
\end{align*}

For a given matrix $H_i = \begin{bmatrix} H_{1,i} & H_{2,i} \end{bmatrix}$, define the set
\[
{\cal L}(H_{i}) = \left\{ (x, x_k) \in {\bf R}^{2 n_x}: \ \Big\vert \ell_j^T
\begin{bmatrix} H_{1,i} & H_{2,i} \end{bmatrix} \begin{bmatrix} x \\ x_k \end{bmatrix} \Big\vert \leq \bar{u}_j, j \in {\bf I}[1,n_u] \right\},
\]
where $\ell_j$ denotes the $j$th column of the identity matrix. According to Lemma~\ref{lem:poly}, within the region $\mathcal{L}(H_i)$ the deadzone nonlinearity \eqref{eqn:LFTPlant} admits the following representation:
\begin{align}
\label{eqn:Comparison}
p_u &= \Delta (u  - H_{1,i} x - H_{2,i} x_k),
\end{align}
for some diagonal matrix $\Delta = {\rm diag}(\delta_1, \ldots, \delta_{n_u})$ with $\delta_j \in [0,1]$.
Consequently, ensuring that the closed-loop trajectories remain within this region leads to the set inclusion condition  \begin{align}
\label{eqn:Inclusion}
\left\{ x_{cl} \in \mathbb{R}^{2 n}: \ x_{cl}^T P_i x_{cl} \leq s^2 \right\} \subseteq {\cal L}(H_{i})
\end{align}
where $P_i$ denotes the Lyapunov matrix associated with the  $i$th subsystem.

The following theorem establishes sufficient conditions for switching stability and weighted $\mathcal{L}_2$-gain performance of the closed-loop system under average dwell-time switching, and provides an explicit LMI-based synthesis procedure for constructing the switching output-feedback controller matrices.

\begin{theorem}
\label{Theorem1}
Consider the ADT switched linear system \eqref{plant} subject to actuator saturation. Given positive scalars $\gamma > 0$, $\lambda_0 > 0$, and $\mu > 1$, suppose there exist positive
definite matrices $R_i, S_i \in \mathbb{S}_+^{n \times n}$,
diagonal matrices $U_i \in \mathbb{S}_+^{n_u \times n_u}$, reset matrices $\hat{\Delta}_{ij}\in\mathbb{R}^{n\times n}$, rectangular matrices $\hat{H}_{1,i},
\hat{H}_{2,i}$ and $\hat{A}_{k,i}, \hat{B}_{k1,i},
\hat{B}_{k2,i}, \hat{C}_{k1,i}, \hat{D}_{k11,i},
\hat{D}_{k12,i}$ of compatible dimensions,
such that for all $i,j \in \mathbf{I}[1,N_p]$ with $i\neq j$, the following conditions hold:
\begin{align}
& \left[ \begin{matrix}
{\it He}\{ A_{p,i} R_i + B_{p2,i} \hat{C}_{k1,i} \} + \lambda_0 R_i & \star & \star \\
\hat{A}_{k,i} + A_{p,i}^T + C_{p2,i}^T \hat{D}_{k11,i}^T B_{p2,i}^T + \lambda_0 I & {\it He}\{S_i A_{p,i} + \hat{B}_{k1,i} C_{p2,i} \} + \lambda_0 S_i \\
-U_i B_{p2,i}^T + \hat{D}_{k12,i}^T B_{p2,i}^T + \hat{C}_{k1,i} - \hat{H}_{2,i} & \hat{B}_{k2,i}^T + \hat{D}_{k11,i} C_{p2,i} - \hat{H}_{1,i} & {\it He} \{ \hat{D}_{k12,i} - U_i \} \\
B_{p1,i}^T + D_{p21,i}^T \hat{D}_{k11,i}^T B_{p2,i}^T & B_{p1,i}^T S_i + D_{p21,i}^T \hat{B}_{k1,i}^T & D_{p21,i}^T \hat{D}_{k11,i}^T \\
C_{p1,i} R_i + D_{p12,i} \hat{C}_{k1,i} & C_{p1,i} + D_{p12,i} \hat{D}_{k11,i} C_{p2,i} & -D_{p12,i} U_i + D_{p12,i} \hat{D}_{k12,i}
\end{matrix} \right. \nonumber \\
& \hspace*{2.5in} \left. \begin{matrix}
\star & \star \\
\star & \star \\
\star & \star \\
-I_{n_w} & \star \\
D_{p11,i} + D_{p12,i} \hat{D}_{k11,i} D_{p21,i} & -\gamma^2 I_{n_z}
\end{matrix} \right] < 0 \label{eqn:LMI1} \\
& \begin{bmatrix}
R_i & \star \\
I_n & S_i
\end{bmatrix} > 0 \label{LMI2} \\
& \begin{bmatrix}
\mu R_i & \star & \star & \star \\
\mu I_n & \mu S_i & \star & \star \\
R_i & I_n & R_j & \star \\
\hat{\Delta}_{ij} & S_j & I_n & S_j
\end{bmatrix} \geq 0 \label{eqn:LMI3} \\
& \begin{bmatrix}
\frac{\bar{u}^2_m}{s^2} & \star & \star \\
{\hat H}_{2,i}^T \ell_m & R_i & \star \\
{\hat H}_{1,i}^T \ell_m & I_n & S_i
\end{bmatrix} \geq 0, \quad \forall \, m \in {\bf I}[1,n_u].  \label{eqn:LMI4}
\end{align}
Then, the closed-loop system (\ref{Closed1})-(\ref{Closed2})
is exponentially stabilized within the region $\bigcap_{i=1}^{N_p} \{ x_{cl}: x_{cl} P_i x_{cl} \leq s^2 \}$ by the hybrid controller (\ref{Controller})
under any switching signal $\sigma(t)$ with average dwell time  $\sigma$ satisfying
\begin{equation}
\label{tau_a}
\tau_a \geq \frac{\ln(\mu)}{\lambda_0}
\end{equation}
and the weighted $\mathcal{L}_2$ gain from disturbance $w$ to output $z$ is less than $\gamma$ for any bounded disturbance $w \in {\cal W}_s$.
Moreover, the hybrid controller matrices \eqref{Controller} can be constructed as follows:
\begin{itemize}
\item
Solve the factorization problem for $M_i, N_i$
\[
I_n - R_i S_i = M_i N_i^T, \qquad i \in {\bf I}[1, N_p].
\]
\item
Compute the controller matrices and reset matrices:
\begin{align}
\begin{aligned}
\label{ABCDDel}
\begin{bmatrix}
A_{k,i} & B_{k2,i} & B_{k1,i} \\
C_{k,i} & D_{k12,i} & D_{k11,i}
\end{bmatrix} &= \begin{bmatrix}
N_i & S_i B_{p2,i} \\
0 & I_{n_u}
\end{bmatrix}^{-1}
\begin{bmatrix}
\hat{A}_{k,i} - S_i A_{p,i} R_i & \hat{B}_{k2,i} + S_i B_{p2,i} U_i & \hat{B}_{k1,i} \\
\hat{C}_{k1,i} & \hat{D}_{k12,i} & \hat{D}_{k11,i}
\end{bmatrix} \\
& \hspace*{0.25in} \times
\begin{bmatrix}
M_i^T & 0 & 0 \\
0 & U_i & 0 \\
C_{p2,i} R_i & 0 & I_{n_y}
\end{bmatrix}^{-1} \\
\Delta_{ij} &= N_j^{-1} (\hat{\Delta}_{ij} - S_j R_i) M_i^{-T}.
\end{aligned}
\end{align}
\item
Compute $H_{1,i}, H_{2,i}$ matrices as
\begin{align}
\begin{bmatrix}
    H_{1,i} & H_{2,i}
\end{bmatrix} = \begin{bmatrix}
    \hat{H}_{1,i} & \hat{H}_{2,i}
\end{bmatrix} \begin{bmatrix}
    I_n & R_i \\
    0 & M_i^T
\end{bmatrix}^{-1}.
\end{align}
\end{itemize}
\end{theorem}

\begin{proof}
Consider the Lyapunov-like function for the closed-loop system (\ref{Closed1})-(\ref{Closed2}) defined as
\begin{align*}
V_i(x_{cl}) &= x_{cl}^T P_i x_{cl} := \begin{bmatrix}
x_p \\
x_k
\end{bmatrix}^T
\begin{bmatrix}
S_i & N_i \\
N_i^T & X_i^{-1}
\end{bmatrix}
\begin{bmatrix}
x_p \\
x_k
\end{bmatrix}.
\end{align*}
Define
\begin{align*}
Z_{1,i} &= \begin{bmatrix}
R_i & I \\
M_i^T & 0
\end{bmatrix}, \qquad
Z_{2,i} = \begin{bmatrix}
I & S_i \\
0 & N_i^T
\end{bmatrix}
\end{align*}
such that $P_i Z_{1,i} = Z_{2,i}$.
Then $X_i^{-1} = -N_i^T R_i M_i^{-T}$.
From condition (\ref{LMI2}),
\[
Z_{1,i}^T P_i Z_{1,i} = \begin{bmatrix}
R_i & I \\
I & S_i
\end{bmatrix} > 0,
\]
which implies $P_i > 0$ for all $i \in {\bf I}[1,N_p]$.

We require
\begin{equation}
\dot{V}_i + \lambda_0 V_i + \frac{1}{\gamma^2} z^T z - w^T w + 2 p_u^T U_i^{-1} (u - H_{i} x_{cl} - p_u) < 0,
\end{equation}
which is equivalent to
\begin{align}
\begin{bmatrix}
{\it He}\{ P_i A_{cl,i} \} + \lambda_0 P_i & \star & \star & \star \\
B_{cl0,i}^T P_i + U_i^{-1} (C_{cl0,i}-H_i) & {\it He} \{ U_i^{-1} (D_{cl00,i} - I) \} & \star & \star \\
B_{cl2,i}^T P_i & D_{cl02,i}^T U_i^{-1} & -I &\star \\
C_{cl2,i} & D_{cl20,i} & D_{cl22,i} & -\gamma^2 I
\end{bmatrix} &< 0 \label{Thm.LMI1*}
\end{align}
Using the congruence transformation ${\rm diag}\{Z_{1,i}, U_i, I, I\}$, this condition (\ref{Thm.LMI1*}) is equivalent to the LMI \eqref{eqn:LMI1}.
Explicit calculations show that all terms in the transformed matrix correspond exactly to the blocks in \eqref{eqn:LMI1} using the relations
\begin{align}
\begin{bmatrix}
    \hat{A}_{k,i} & \hat{B}_{k2,i} & \hat{B}_{k1,i} \\
    \hat{C}_{k1,i} & \hat{D}_{k12,i} & \hat{D}_{k11,i}
\end{bmatrix} &= \begin{bmatrix}
	S_i A_{p,i} R_i & -S_i B_{p2,i} U_i & 0 \\
	0 & 0 & 0
\end{bmatrix} \nonumber \\
&\hspace{-1.0in} + \begin{bmatrix}
    N_i & S_i B_{p2,i} \\
    0 & I
\end{bmatrix} \begin{bmatrix}
    A_{k,i} & B_{k2,i} & B_{k1,i} \\
    C_{k1,i} & D_{k12,i} & D_{k11,i}
\end{bmatrix} \begin{bmatrix}
	M_i^T & 0 & 0 \\
	0 & U_i & 0 \\
	C_{p2,i} R_i & 0 & I
\end{bmatrix}. \label{eqn:temp1}
\end{align}
This is demonstrated below:
\begin{align*}
Z_{1,i}^T P_i A_{cl,i} Z_{1,i} &= Z_{2,i}^T A_{cl,i} Z_{1,i} \\
&= \begin{bmatrix}
A_{p,i} R_i + B_{p2,i} \hat{C}_{k1,i} & A_{p,i} + B_{p2,i} \hat{D}_{k11,i} C_{p2,i} \\
\hat{A}_{k,i} & S_i A_{p,i} + \hat{B}_{k1,i} C_{p2,i}
\end{bmatrix} \\
U_i \left[ B_{cl0,i}^T P_i + U_i^{-1} (C_{cl0,i} - H_i) \right] Z_{1,i} &= \\
& \hspace*{-0.5in} \begin{bmatrix}
    -U_i B_{p2,i}^T + \hat{D}^T_{k12,i} B_{p2,i}^T +  \hat{C}_{k1,i} - \hat{H}_{2,i} & \hat{B}_{k2,i}^T + \hat{D}_{k11,i} C_{p2,i} - \hat{H}_{1,i}
\end{bmatrix} \\
U_i \left[ U_i^{-1} (D_{cl00,i} -I) + (D_{cl00,i} - I)^T U_i^{-1} \right] U_i &= \hat{D}_{k12,i} - U_i + \hat{D}^T_{k12,i} - U_i \\
B_{cl2,i}^T P_i Z_{1,i} &= B_{cl2,i}^T Z_{2,i} \\
&= \begin{bmatrix}
B_{p1,i}^T + D_{p21,i}^T \hat{D}_{k11,i}^T B_{p2,i}^T & B_{p1,i}^T S_i + D_{p21,i}^T \hat{B}_{k1,i}^T
\end{bmatrix} \\
(D_{cl02,i}^T U_i^{-1}) U_i &= D_{p21,i}^T \hat{D}_{k11,i}^T \\
C_{cl2,i} Z_{1,i} &= \begin{bmatrix}
C_{p1,i} R_i + D_{p12,i} \hat{C}_{k1,i} & C_{p1,i} + D_{p12,i} \hat{D}_{k11,i} C_{p2,i}
\end{bmatrix} \\
D_{cl20,i} U_i &= -D_{p12,i} U_i + D_{p12,i} \hat{D}_{k12,i}  \\
D_{cl22,i} &= D_{p11,i} + D_{p12,i} \hat{D}_{k11,i} D_{p21,i}
\end{align*}

At a switching instant from subsystem $i$ to $j$,
$V_j(x_{cl}^+) = {x_{cl}^+}^T P_j x_{cl}^+ = x_{cl}^T A_{s,ij}^T P_j A_{s,ij} x_{cl}$.
The boundary condition $V_j(x_{cl}^+) \le \mu V_i(x_{cl})$ is equivalent, via the Schur complement and congruence with ${\rm diag} \{Z_{1,i}, Z_{1,j}\}$, to the LMI \eqref{eqn:LMI3} with
\begin{equation}
\label{Thm.barDel}
\hat{\Delta}_{ij} = S_j R_i + N_j \Delta_{ij} M_i^T.
\end{equation}
This can be shown as follows:
\begin{gather}
\mu V_i(x_{cl}) - V_j(x_{cl}^+) =
x_{cl}^T (\mu P_i - A_{s,ij}^T P_j A_{s,ij}) x_{cl} \geq 0 \label{Thm.LMI2*} \\
\Leftrightarrow \qquad \begin{bmatrix}
\mu P_i & \star \\
P_j A_{s,ij} & P_j
\end{bmatrix} \geq 0. \nonumber
\end{gather}
Note that
\begin{align*}
Z_{1,j}^T P_j A_{s,ij} Z_{1,i} &= Z_{2,j}^T A_{s,ij} Z_{1,i} = \begin{bmatrix}
R_i & I \\
\hat{\Delta}_{ij} & S_j
\end{bmatrix}
\end{align*}
The controller and reset matrices in \eqref{ABCDDel} follow directly by inverting the relations (\ref{eqn:temp1}) and (\ref{Thm.barDel}).

From Lemma \ref{lem:poly}, the set inclusion condition
\eqref{eqn:Inclusion}
ensures that the deadzone representation of the saturation is valid.
Using the congruence transformation with
$Z_i = {\rm diag} \left\{ 1,
\begin{bmatrix} R_i & M_i \\ I & 0 \end{bmatrix} \right\}$, we recover LMI \eqref{eqn:LMI4} and the relations
\begin{equation}
\begin{bmatrix}
    \hat{H}_{1,i} & \hat{H}_{2,i}
\end{bmatrix} = \begin{bmatrix}
    H_{1,i} & H_{2,i}
\end{bmatrix} \begin{bmatrix}
    I & R_i \\
    0 & M_i^T
\end{bmatrix}. \label{eqn:temp2}
\end{equation}

Hence, the LMIs \eqref{eqn:LMI1}–\eqref{eqn:LMI4} guarantee exponential stability under average dwell-time switching and a weighted $\mathcal{L}_2$ gain $\gamma$ \cite{ZhaHYM.JFI01}.
The controller construction in \eqref{ABCDDel} and reset matrices $\Delta{ij}$ satisfy the theorem.
\end{proof}


\section{Illustrating Example}
\label{Sec.Sim}

In this section, a numerical example adapted from the study in \cite{YuaW2015} is presented to demonstrate the effectiveness of the proposed hybrid saturated control scheme.


We consider a switched linear system with $N_p = 2$ subsystems given below
\begin{align}
\label{Sim.plant}
\begin{aligned}
A_{p,1} &= \begin{bmatrix}
0.5108 & -0.9147 & -0.2 \\
-0.6563 & 0.1798 & 0.113 \\
0.881 & -0.7841 & 0.1
\end{bmatrix}, \quad
A_{p,2} = \begin{bmatrix}
-0.125 & -0.9833 & -0.34 \\
-0.5305 & 0.3848 & 0.58 \\
1.0306 & 0.6521 & 0.1
\end{bmatrix} \\
B_{p1,1} &= \begin{bmatrix}
0.1056 \\
0.1284 \\
0.1
\end{bmatrix}, \quad
B_{p1,2} = \begin{bmatrix}
0.7425 \\
0.1436 \\
0.1 \\
\end{bmatrix}, \\
B_{p2,1} &= \begin{bmatrix}
0.3257 \\
1.2963 \\
2.43
\end{bmatrix}, \quad
B_{p2,2} = \begin{bmatrix}
1.0992 \\
0.6532 \\
3.5
\end{bmatrix} \\
C_{p1,1} &= \begin{bmatrix}
0.01 & 0.06 & 0.03
\end{bmatrix}, \quad
C_{p1,2} = \begin{bmatrix}
0.01 & 0.02 & 0.05
\end{bmatrix} \\
C_{p2,1} &= \begin{bmatrix}
-5 & 0.2 & 0.5
\end{bmatrix}, \quad
C_{p2,2} = \begin{bmatrix}
-6 & 6 & -1
\end{bmatrix} \\
D_{p11,1} &= D_{p11,2} = 0, \quad D_{p12,1} = D_{p12,2} = 0,
\quad D_{p21,1} = D_{p21,2} = 0.1.
\end{aligned}
\end{align}
As noted in \cite{YuaW2015}, these subsystems do not admit a common Lyapunov function. Consequently, the switched plant cannot be stabilized under arbitrary switching sequences. However, as will be shown, this stabilization problem can be effectively addressed using the proposed hybrid saturation  control framework combined with the average dwell-time (ADT) switching strategy.

For the saturation control design, the disturbance energy level is chosen as $s = 0.42$.
The switched plant (\ref{Sim.plant}) is simulated over the time interval $t\in [0,70] ~sec$ under the following switching sequence.
The switching signal for the first $2 sec$ is at subsystem 2. Over the interval $[2,50] ~sec$, the switching signal follows a cyclic pattern:
\begin{equation}
\label{Sim.swlogic}
\sigma = \sigma_{cyclic} := \left\{ \begin{split}
1, & \ 2lT+2 \leq t \leq (2l+1)T+2 \\
2, & \ (2l+1)T+2 < t\leq (2l + 2)T+2
\end{split} \right.
\end{equation}
with $l = 0, 1, 2, 3, \cdots$, and $T = 12 ~sec$.
After $50 sec$, the switching signal remains at subsystem 1.
(see Fig. \ref{Fig.sigma}).
Under this switching logic, six switching events occur over the interval $[0,70]$, yielding an average dwell time
\[
\tau_a = 70/5 = 14.0 (sec).
\]

\begin{figure}[htb]
\centering
\includegraphics[width=0.4\textwidth]{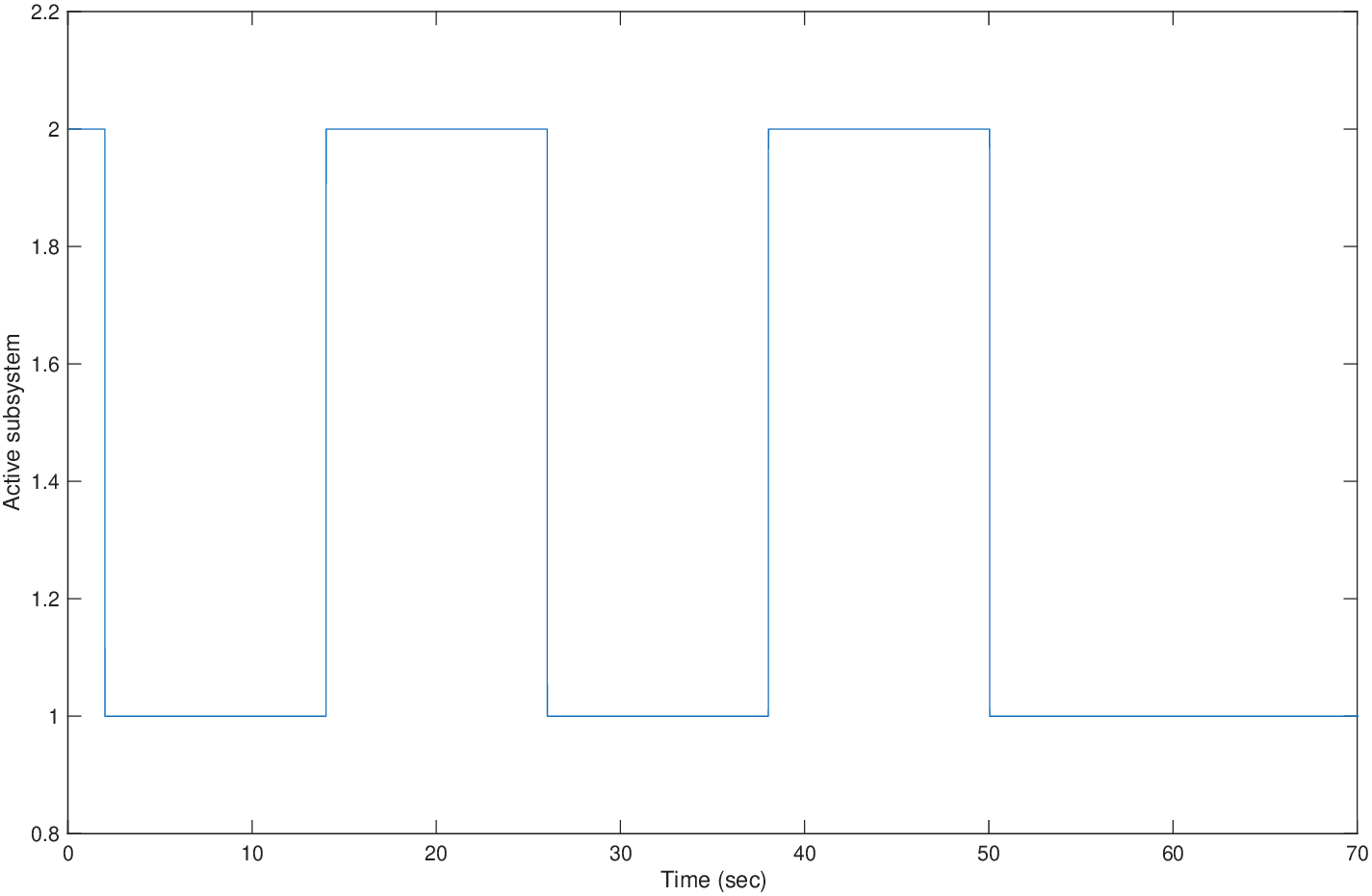}
\caption{Switching signal for the switched plant (\ref{Sim.plant}).}
\label{Fig.sigma}
\end{figure}

For ADT switching, the weighted ${\cal L}_2$ gain performance is influenced by the dwell-time parameters
$\lambda_0$ and $\mu$.
Due to the presence of actuator saturation, it is expected that the achievable performance of the hybrid saturated controller will not surpass that of the unsaturated design reported in \cite{YuaW2015}.
Table \ref{Table.delta} summarizes the optimized values of the performance index
$\gamma$ obtained from Theorem \ref{Theorem1} for different combinations of these parameters.
The results illustrate the inherent trade-off between disturbance attenuation and dwell-time requirements.
When $\lambda_0 = 0.1$ and $\mu$ varies from $3.8$ to $4.2$,
the required dwell time $\tau^*_a$ increases, while the achievable performance level $\gamma$ improves.
Conversely, for fixed $\mu$, increasing $\lambda_0$
reduces $\tau^*_a$ but leads to a degradation in the optimized performance.
These observations confirm that larger values of
$\mu$ improve disturbance attenuation at the cost of more restrictive switching constraints, whereas variations in
$\lambda_0$ produce the opposite effect.
Moreover, for a prescribed $\lambda_0$, feasibility of the synthesis conditions \eqref{eqn:LMI1}-\eqref{eqn:LMI4} depends on the admissible range of $\mu$.
If $\mu$ is chosen too small, these conditions become infeasible due to the switching constraint and set-inclusion conditions associated with actuator saturation.

\begin{table}[htb]
\caption{Effects of constants $\lambda_0$ and $\mu$ on optimal weighted $\mathcal{L}_2$-gain performance.}
\label{Table.delta}
\begin{center}
\begin{tabular}{lcccc} \hline \hline
$\lambda_0$ & $\mu$ & $\tau_a^* = \frac{\ln(\mu)}{\lambda_0} sec$ & $\gamma$ \\ \hline
0.05 & 3.4 & 24.476 & 1.7017 \\
0.05 & 3.8 & 26.70 & 0.4574 \\
0.05 & 4.2 & 28.702 & 0.3135 \\
0.1 & 3.8 & 13.35 & 2.047 \\
0.1 & 4.2 & 14.35 & 0.4951 \\
0.1 & 4.6 & 15.261 & 0.3368 \\
0.12 & 4 & 11.55 & 1.5055 \\
0.12 & 4.4 & 12.347 & 0.4890 \\ \hline \hline
\end{tabular}
\end{center}
\end{table}

To demonstrate the proposed approach in time-domain simulation,  we select $\lambda_0 = 0.1$ and $\mu
= 4$, representing a reasonable balance between dwell-time requirements and disturbance attenuation performance.
With these choices, the weighted ${\cal L}_2$ performance is $\gamma = 0.6953$ and the required average dwell time obtained from condition (\ref{tau_a}) is $\tau_a^* = 13.86 ~sec$, which is satisfied by the implemented switching signal (\ref{Sim.swlogic}).
The synthesized controller gains using Theorem
\ref{Theorem1} are
\begin{align*}
A_{k,1} &= \begin{bmatrix}
  -3.6451 & -2.2847 & 0.093568 \\
   6.9684 & -8.6748 & -1.3174 \\
   107.72 & -27.103 & -12.288
  \end{bmatrix}, \quad \begin{bmatrix} B_{k1,1} & B_{k2,1} \end{bmatrix} = \begin{bmatrix}
  -1.1492 & -8.3909 \times 10^3 \\
  -0.21955 & -0.049538 \\
   16.910 & -0.13303
  \end{bmatrix} \\
C_{k1,1} &= \begin{bmatrix}
3.8387 & -6.3871 & -0.87224
\end{bmatrix}, \quad \begin{bmatrix} D_{k11,1} & D_{k12,1} \end{bmatrix} = \begin{bmatrix}
-0.50698 & 0.96418
\end{bmatrix} \\
A_{k,2} &= \begin{bmatrix}
  -1.3183 \times 10^5 & 1.6129\times 10^5 & 298.99 \\
   1.7675\times 10^5 & -2.1627\times 10^5 & -408.46 \\
   1.0858\times 10^6 & -1.3286\times 10^6 & -2.5270\times 10^3
  \end{bmatrix}, \\
\begin{bmatrix} B_{k1,2} & B_{k2,2} \end{bmatrix} &= \begin{bmatrix}
     -2.0805\times 10^4 & 329.37 \\
   2.7895\times 10^4 & -441.71 \\
   1.7136\times 10^5 & -2.7137\times 10^3
  \end{bmatrix} \\
C_{k1,2} &= \begin{bmatrix}
 -6.6978 & 3.3353 & -3.8006
\end{bmatrix}, \quad \begin{bmatrix} D_{k11,2} & D_{k12,2} \end{bmatrix} =
\begin{bmatrix}
-1.3190 & 0.96738
\end{bmatrix}
\end{align*}
with
\begin{align*}
\Delta_{12} &= \begin{bmatrix}
   0.90789 & 0.08393 & 4.2385\times 10^{-3} \\
   0.10442 & 0.90187 & -3.9851\times 10^{-3} \\
  -0.054297 & 0.2453 & 0.94581
  \end{bmatrix}, \quad \Delta_{21} = \begin{bmatrix}
   0.95265 &  0.16426 & 0.088899 \\
   0.063434 & 0.77962 & -0.11918 \\
   0.39005 & -1.3532 & 0.26763
\end{bmatrix}.
\end{align*}
An important feature of the proposed reset mechanism is that the reset matrix $\Delta_{ij}$ depends explicitly on both the pre-switching and post-switching subsystems.
This contrasts with existing reset-based designs, where resets are typically determined solely by the post-switching subsystem \cite{LuW.TCST06}.
This structure provides greater flexibility and leads to improved transient behavior.

Applying the hybrid saturated controller, simulations are conducted with zero initial conditions and a pulsed disturbance of magnitude 0.6 applied for $0.4 sec$.
The results in Fig. \ref{Fig.Full} indicate that the closed-loop system remains stable.
Notably, the controller state resets at switching instants effectively reduce excessive control effort during subsystem transitions.

\begin{figure}[htb]
\centering
\subfloat[Controlled output]{
\includegraphics[width=0.45\textwidth]{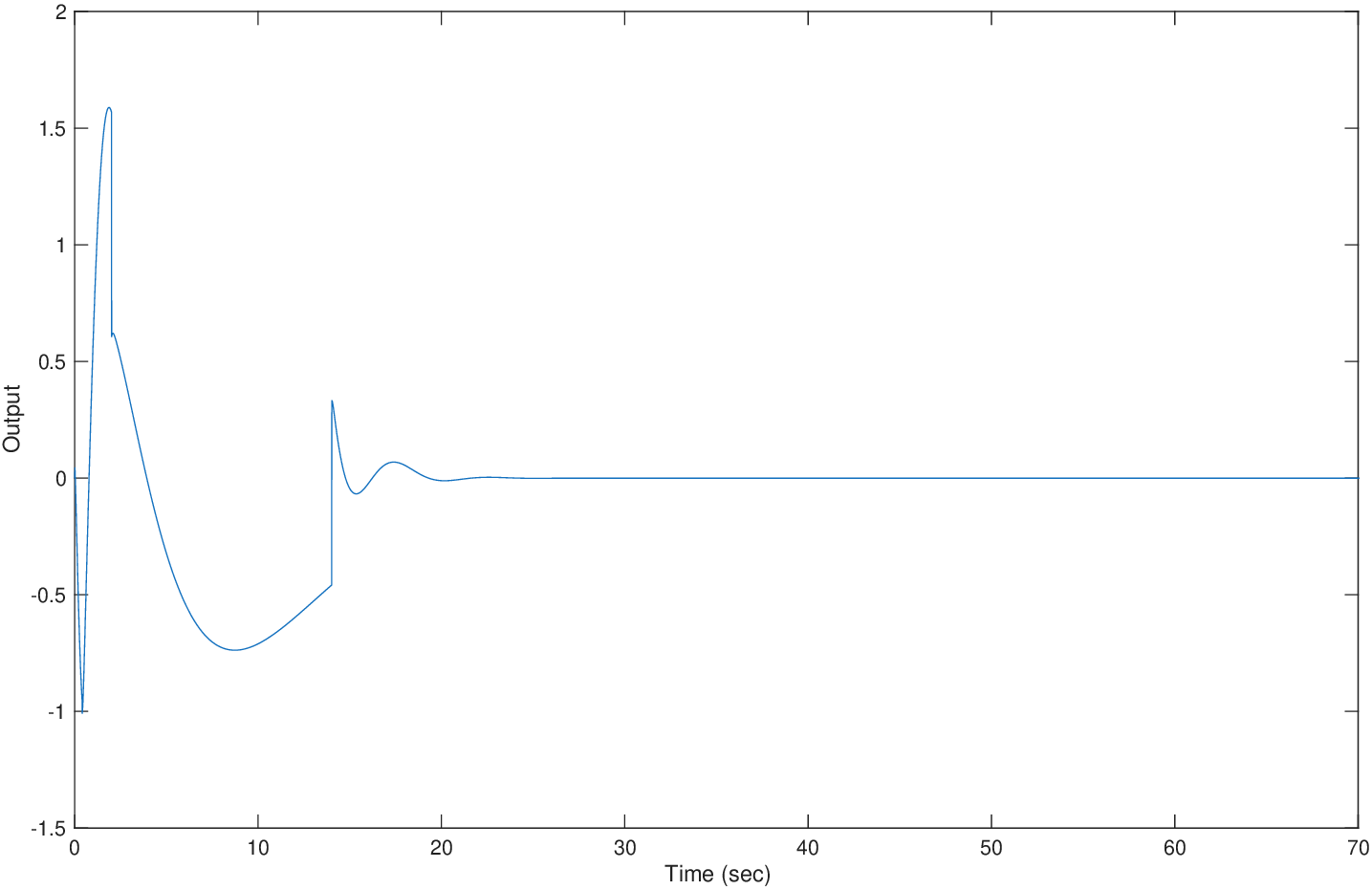}
}
\subfloat[Control input]{
\includegraphics[width=0.45\textwidth]{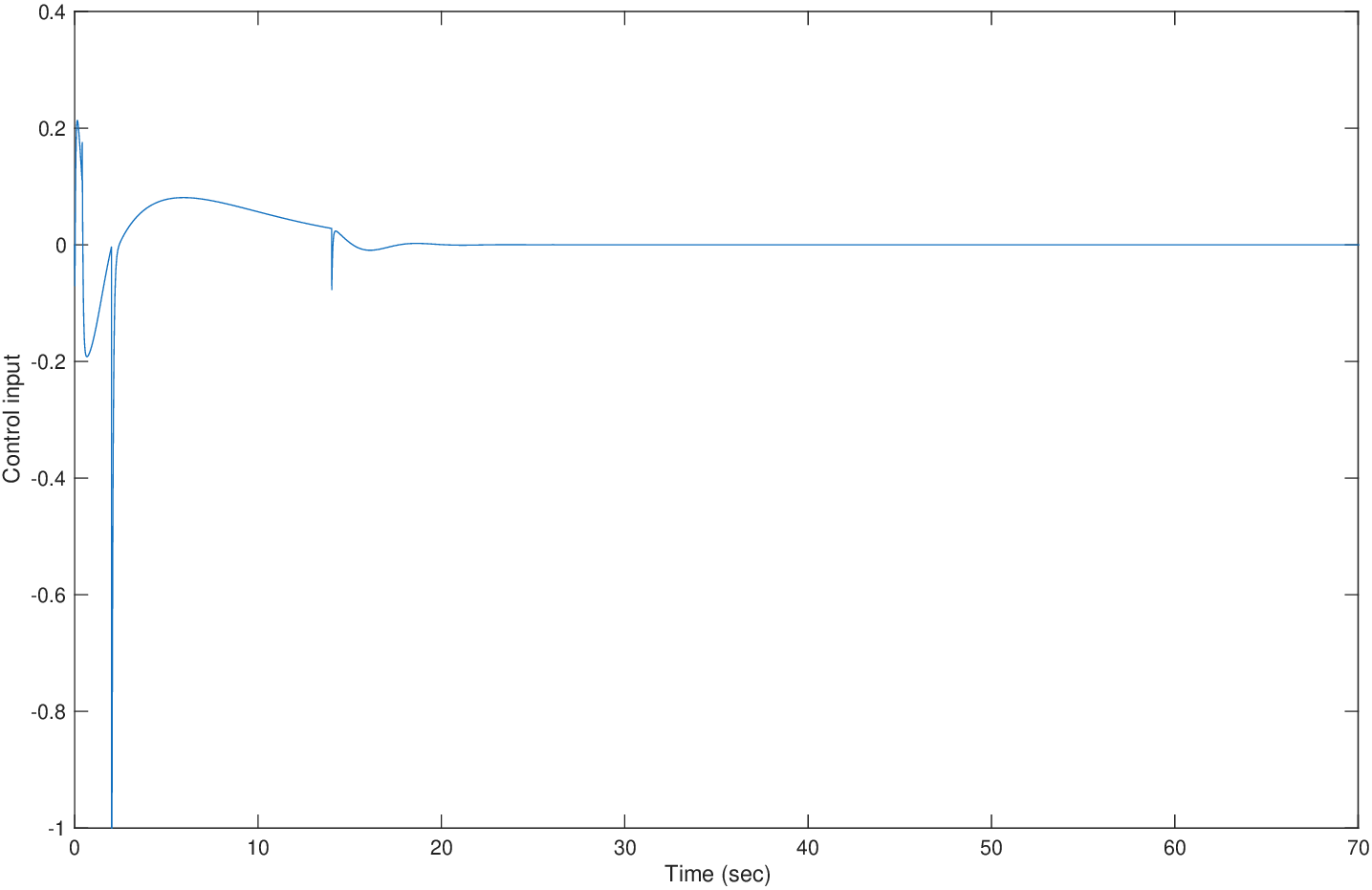}
} \\
\subfloat[Plant states]{
\includegraphics[width=0.45\textwidth]{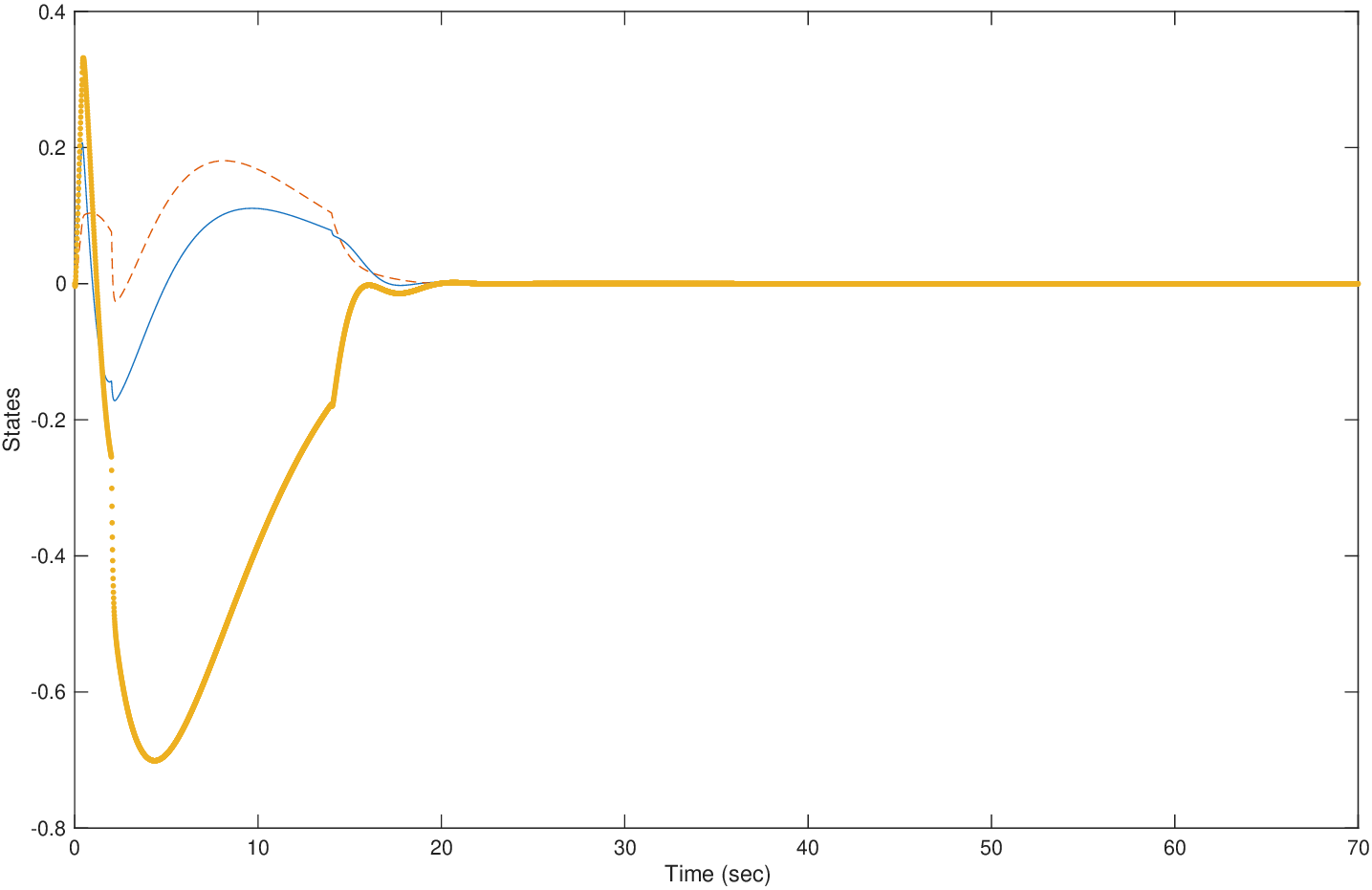}
}
\caption{Simulation results}
\label{Fig.Full}
\end{figure}


Overall, both the numerically optimized weighted
${\cal L}_2$-gain results and the time-domain simulations demonstrate that the proposed hybrid saturation control  strategy offers a practical and effective means of stabilizing saturated switched systems while maintaining reasonable control effort.

\begin{remark}
The synthesis conditions \eqref{eqn:LMI1}-\eqref{eqn:LMI4} may yield excessively large controller gains, particularly in matrices $B_{k1,i}$, $C_{k1,i}$ and $D_{k12,i}$.
To improve transient behavior and mitigate large control discontinuities at switching instants, the hybrid controllers in this example were synthesized by incorporating additional LMI constraints, including $\mathcal{H}_\infty$ performance objectives and controller gain bounds.
This multi-objective design strategy proved effective in obtaining feasible saturation controllers without degrading closed-loop performance.
Compared with pole placement approaches \cite{ChilaliG.TAC96}, this formulation offers improved numerical tractability and systematic tuning flexibility.
Importantly, all design objectives can be incorporated within the LMI framework.
\end{remark}

\section{Conclusions}
\label{Sec.Con}

This paper has presented a hybrid output-feedback control framework for average dwell-time switched linear systems subject to actuator saturation.
By explicitly modeling saturation nonlinearities via a deadzone representation and embedding a reset mechanism into the controller dynamics, the proposed approach provides a unified solution to stability and performance analysis for constrained switched systems, including those with exponentially unstable open-loop subsystems.

The primary contribution of this work is the development of a systematic hybrid control synthesis method that enables dynamic output-feedback controller design under actuator saturation using convex LMI conditions.
Through the introduction of controller state resets and carefully constructed Lyapunov-like functions, the nonconvex boundary conditions inherent to ADT switching and set-inclusion condition for saturation control are incorporated directly into a convex optimization framework.
As a result, exponential stability and weighted ${\cal L}_2$-gain disturbance attenuation guarantees are obtained in a computationally efficient manner.
An explicit controller reconstruction algorithm is also provided, allowing practical implementation of the proposed hybrid controllers.

In addition to the theoretical developments, a numerical example is presented to demonstrate the effectiveness of the proposed hybrid saturation control scheme.
The example illustrates the ability of the method to stabilize a saturated ADT switched system and achieve the desired disturbance attenuation performance, thereby validating both the feasibility and practical advantages of the proposed design approach.

Overall, the results of this paper extend existing hybrid and switched control methodologies to a broader class of constrained systems and provide a tractable design framework for output-feedback control of ADT switched linear systems with actuator saturation.


\begin{thebibliography}{10}
\providecommand{\url}[1]{#1}
\csname url@samestyle\endcsname
\providecommand{\newblock}{\relax}
\providecommand{\bibinfo}[2]{#2}
\providecommand{\BIBentrySTDinterwordspacing}{\spaceskip=0pt\relax}
\providecommand{\BIBentryALTinterwordstretchfactor}{4}
\providecommand{\BIBentryALTinterwordspacing}{\spaceskip=\fontdimen2\font plus
\BIBentryALTinterwordstretchfactor\fontdimen3\font minus
  \fontdimen4\font\relax}
\providecommand{\BIBforeignlanguage}[2]{{%
\expandafter\ifx\csname l@#1\endcsname\relax
\typeout{** WARNING: IEEEtran.bst: No hyphenation pattern has been}%
\typeout{** loaded for the language `#1'. Using the pattern for}%
\typeout{** the default language instead.}%
\else
\language=\csname l@#1\endcsname
\fi
#2}}
\providecommand{\BIBdecl}{\relax}
\BIBdecl

\bibitem{Ant.IEEE00}
P.~J. Antsaklis,
``A brief introduction to the theory and applications of hybrid systems,''
\emph{Proc. {IEEE}}, 88(7):879--887, 2000.

\bibitem{BanW2015}
X. Ban and F. Wu,
``Gain scheduling output feedback control of linear plants with actuator saturation,''
{\em Journal of Franklin Institute}, 352(10):4163-4187, 2015. DOI: 10.1016/j.jfranklin.2015.06.005

\bibitem{BerM95}
D.S. Bernstein and A.N. Michel,
``A chronological bibliography on saturating actuators,''
{\em Int. J. Robust Non. Contr.}, 5:375-380, 1995.

\bibitem{BoyGFB.B04}
S.~Boyd, L.~E. Ghaoui, E.~Feron, and V.~Balakrishnan, \emph{Linear Matrix Inequalities in System and Control Theory}, Philadelphia, PA: SIAM, 2004.

\bibitem{Bra.TAC98}
M.~S. Branicky,
``Multiple {Lyapunov} functions and other analysis tools for switched and hybrid systems,''
\emph{{IEEE} Trans. Autom. Control}, 43(4):475-482, 1998.

\bibitem{BraBM.TAC98}
M.~S. Branicky, V.~S. Borkar, and S.~K. Mitter,
``A unified framework for hybrid control: Model and optimal control theory,''
\emph{{IEEE} Trans. Autom. Control}, 43(1):31-45, 1998.


\bibitem{CheM2021}
Q. Chen and R. Ma,
``Dwell-time-based robust stabilization and
$L_2$-gain analysis for asynchronously switched
linear systems with saturation,''
{\em IEEE Trans. Automation Science Eng.}, 21(4):6882-6891, 2024.

\bibitem{ChilaliG.TAC96}
M.~Chilali and P.~Gahinet,
``$\mathcal{H}_\infty$ design with pole placement constraints: an {LMI} approach,''
\emph{{IEEE} Trans. Autom. Control}, 41(3):358-367, 1996.



\bibitem{DaiHTZ2009}
D. Dai, T. Hu, A.R. Teel, and L. Zaccarian,
``Output feedback design for saturated linear plants
using deadzone loops,''
{\em Automatica}, {\em 45}:2917-2924, 2009.


\bibitem{DuaW2016}
C. Duan and F. Wu,
``New results on switched linear systems with actuator saturation,''
{\em International Journal of Systems Science}, 47(5):1008-1020, 2016. DOI: 10.1080/00207721.2014.911386




\bibitem{GoeST2012}
R. Goebel, R.G. Sanfelice, and A.R. Teel,
{\em Hybrid Dynamical Systems},
Princeton, NJ: Princeton University Press, 2012.




\bibitem{HuL2001}
T. Hu and Z. Lin,
{\em Control Systems with Actuator Saturation: Analysis
and Design},
Boston, MA: Birkhauser, 2001.

\bibitem{HuTZ2006}
T. Hu, A.R. Teel, and L. Zaccarian,
``Stability and performance for saturated systems via
quadratic and non-quadratic Lyapunov functions,''
{\em IEEE Trans. Autom. Control}, 51(11):1770-1786, 2006.




\bibitem{LiberzonB03}
D.~Liberzon,
\emph{Switching in Systems and Control},
Boston, MA: Birkhauser, 2003.

\bibitem{LibM.CSM99}
D.~Liberzon and A.~S. Morse,
``Basic problems in stability and design of switched systems,'' \emph{IEEE Control Systems Magazine}, 19(5):59-70, 1999.

\bibitem{LinA.TAC09}
H.~Lin and P.~J. Antsaklis,
``Stability and stabilizability of switched linear systems: A survey of recent results,''
\emph{{IEEE} Trans. Autom. Control}, 54(1):308-322, 2009.

\bibitem{LinST96}
Z. Lin, A. Saberi, and A.R. Teel,
``Almost disturbance decoupling with internal stability
for linear systems subject to input saturation--state
feedback case,''
{\em Automatica}, 32(4):619-624, 1996.

\bibitem{LiuCS96}
W. Liu, Y. Chitour, and E. Sontag,
``On finite gain stabilizability of linear systems subject to input saturation,''
{\em SIAM J. Contr. Optimization}, 34(4):1190-1219, 1996.

\bibitem{LuW.Au04}
B.~Lu and F.~Wu,
``Switching {LPV} control designs using multiple parameter-dependent {Lyapunov} functions,''
\emph{Automatica}, 40:1973-1980, 2004.

\bibitem{LuW.TCST06}
B.~Lu, F.~Wu, and S.~Kim,
``Switching {LPV} control of an {F-16} aircraft via controller state reset,''
\emph{IEEE Trans. Control Syst. Techno.}, 14(2):267-277, 2006.

\bibitem{McCK.IEEE00}
N.~H. McClamroch and I.~Kolmanovsky,
``Performance benefits of hybrid control design for linear and nonlinear systems,''
\emph{Proc. of IEEE}, 88(7):1083-1096, 2000.

\bibitem{Mor.TAC96}
A.~S. Morse,
``Supervisory control of families of linear set-point controllers--Part 1: exact matching,''
\emph{{IEEE} Trans. Autom. Control}, 41(10):1413-1431, 1996.

\bibitem{Mor.B97}
A.~S. Morse, \emph{Control using Logic-based Switching}, Heidelberg, Germany: Springer, 1997.

\bibitem{NguJ99}
T. Nguyen and F. Jabbari,
``Disturbance attenuation for systems with input
saturation: An LMI approach,''
{\em IEEE Trans. Autom. Control}, 44(4):852-857, 1999.

\bibitem{PaiTGC2002}
C. Paim, S. Tarbouriech, J.M. Gomes da Silva Jr., and
E.B. Castelan,
``Control design for linear systems with saturating
actuators and $\ell_2$-bounded disturbances,''
in {\em Proc. 41st IEEE Conf. Dec. Contr.}, pp.~4148-4153, 2002.


\bibitem{SabLT96}
A. Saberi, Z. Lin, and A.R. Teel,
``Control of linear systems with saturating actuators,''
{\em IEEE Trans. Automat. Contr.}, 41(3):368-378, 1996.

\bibitem{TarG97}
S. Tarbouriech, and G. Garcia (eds),
{\em Control of Uncertain Systems with Bounded Inputs}.
London, UK: Springer, 1997.

\bibitem{SchS.B00}
A.~van~der Schaft and H.~Schumacher,
\emph{An Introduction to Hybrid Dynamical Systems},
Springer, 2000.

\bibitem{ScoFE2002}
G. Scorletti, J.P. Folcher, and L. {El Ghaoui},
``Output feedback control with input saturations: LMI
design approaches,''
{\em Euro. J. Contr.}, 7(6):567-579, 2002.

\bibitem{SunG.B05}
Z.~Sun and S.~S. Ge,
\emph{Switched Linear Systems: Control and Design},
Verlag, NY: Springer, 2005.

\bibitem{SusSY94}
H.J. Sussmann, E.D. Sontag, and Y. Yang,
``A general result on the stabilization of linear systems using bounded controls,''
{\em IEEE Trans. Autom. Control}, 39(12):2411-2425, 1994.

\bibitem{WanZDL.IJRNC09}
R.~Wang, J.~Zhao, G.~M. Dimirovski, and G.~P. Liu,
``Output feedback control for uncertain linear systems with faulty actuators based on a switching method,''
\emph{Int. J. Robust Nonl. Control}, 19:1295-1312, 2009.

\bibitem{WicksPD.EJC98}
M.~A. Wicks, P.~Peletics, and R.~A. Decarlo,
``Switched controller synthesis for the quadratic stabilization of a pair of unstable linear systems,''
\emph{Eur. J. Control}, 4:140-147, 1998.

\bibitem{Wu2001}
F. Wu,
``Switching LPV control design for magnetic bearing systems,''
in {\em Proc. IEEE Conference on Control Applications}, pp. 41-47, 2001. DOI: 10.1109/CCA.2001.973835

\bibitem{WuLZ2007}
F. Wu, Z. Lin, and Q. Zheng,
``Output feedback stabilization of linear systems with actuator saturation,''
{\em IEEE Trans. Autom. Control}, 52(1):122-128, 2007.

\bibitem{WuZL2009}
F. Wu, Q. Zheng, and Z. Lin,
``Disturbance attenuation by output feedback for linear systems subject to actuator saturation,''
{\em International Journal of Robust and Nonlinear Control}, 19:168-184, 2009. DOI: 10.1002/rnc.1306

\bibitem{YeMH.TAC98}
H.~Ye, A.~N. Michel, and L.~Hou,
``Stability theory for hybrid dynamical systems,''
\emph{{IEEE} Trans. Autom. Control}, 43(4):461-474, 1998.

\bibitem{YuaW2015}
C. Yuan and F. Wu,
``Hybrid control for switched linear systems with
average swell time,''
{\em IEEE Trans. Automat. Control}, 60(1):240-245, 2015.

\bibitem{YuaLWD2016}
C. Yuan, Y. Liu, F. Wu, and C. Duan,
``Hybrid switched gain-scheduling control for missile
autopilot design,''
{\em Journal of Guiance, Control and Dynamics}, 39(10):2352-2363, 2016. DOI: 10.2514/1.G001791


\bibitem{ZhaHYM.JFI01}
G.~Zhai, B.~Hu, K.~Yasuda, and A.~N. Michel,
``Disturbance attenuation properties of time-controlled switched systems,''
\emph{J. Franklin Institute}, 338:765-779, 2001.

\bibitem{ZhoY2015}
G.-X. Zhong and G.-H. Yang,
``${\cal L}_2$-gain analysis and control of saturated switched systems with a dwell time constraint,''
{\em Nonlinear Dyn.}, 80:1231-1244, 2015.



\end{thebibliography}

\end{document}